\documentclass[11pt,letterpaper]{article}
\usepackage{hyperref}
\usepackage{epsfig,amssymb,amsmath,color}
\usepackage{multirow}
\usepackage{pictexwd,dcpic}
\usepackage{afterpage}
\usepackage[headheight=-1.5in,margin=0.75in]{geometry}
\newcommand{\includefigs}[1]{}

\newcommand{\mnpq}[4]{\langle #1,#2,#3\rangle=#4}

\newcommand{\bmat}{\left[ \begin{array}}
\newcommand{\emat}{\end{array} \right]}

\newtheorem{thm}{Theorem}
\newtheorem{theorem}[thm]{Theorem}
\newtheorem{lemma}[thm]{Lemma}
\newtheorem{corollary}[thm]{Corollary}

\newtheorem{proposition}[thm]{Proposition}
\newtheorem{fact}[thm]{Fact}

\newcommand{\qed}{\hfill \mbox{\raggedright \rule{.07in}{.1in}}}
\newenvironment{proof}{\vspace{1ex}\noindent{\bf Proof}\hspace{0.5em}}
	{\hfill\qed\vspace{1ex}}
\newenvironment{pfof}[1]{\vspace{1ex}\noindent{\bf Proof of #1}\hspace{0.5em}}
	{\hfill\qed\vspace{1ex}}

\newcommand{\remove}[1]{}

\newcommand{\lt}{\left}
\newcommand{\rt}{\right}

\graphicspath{{figures/}}

\newcommand{\ignore}[1]{}

\begin{document}

\title{Graph Expansion Analysis for Communication Costs of
Fast Rectangular Matrix Multiplication}
\author{
Grey Ballard
\thanks{EECS Department,
University of California, Berkeley, CA 94720.
(ballard@eecs.berkeley.edu).
Research supported by Microsoft (Award $\#$024263) and Intel (Award $\#$024894) funding and by matching funding by U.C. Discovery (Award $\#$DIG07-10227). Additional support comes from Par Lab affiliates National Instruments, Nokia, NVIDIA, Oracle, and Samsung.} 
\and James Demmel
\thanks{Mathematics Department and CS Division,
University of California, Berkeley, CA 94720.
(demmel@cs.berkeley.edu).
Research supported by Microsoft (Award $\#$024263) and Intel (Award $\#$024894) funding and by matching funding by U.C. Discovery (Award $\#$DIG07-10227). Additional support comes from Par Lab affiliates National Instruments, Nokia, NVIDIA, Oracle, and Samsung.
Research is also supported by DOE grants DE-SC0003959, DE- SC0004938, and DE-AC02-05CH11231.}
\and Olga Holtz
\thanks{Departments of Mathematics,
University of California, Berkeley and Technische Universit\"at
Berlin. 
(holtz@math.berkeley.edu)
Research supported by the Sofja Kovalevskaja
programme of Alexander von Humboldt Foundation and by the
National Science Foundation under agreement DMS-0635607.} 
\and Benjamin Lipshitz
\thanks{EECS Department,
University of California, Berkeley, CA 94720.
(lipshitz@berkeley.edu)
Research supported by Microsoft (Award $\#$024263) and Intel (Award $\#$024894)
funding and by matching funding by U.C. Discovery (Award $\#$DIG07-10227).} 
\and Oded Schwartz
\thanks{EECS Department,
University of California, Berkeley, CA 94720.
(odedsc@eecs.berkeley.edu)
Research supported by U.S. Department of Energy grants under
Grant Numbers DE-SC0003959.}
}
\maketitle

\abstract{Graph expansion analysis of computational DAGs is useful for obtaining communication cost lower bounds where previous methods, such as geometric embedding, are not applicable. This has recently been demonstrated for Strassen's and Strassen-like fast square matrix multiplication algorithms.  Here we extend the expansion analysis approach to fast algorithms for rectangular matrix multiplication, obtaining a new class of communication cost lower bounds.  These apply, for example to the algorithms of Bini et~al.~(1979) and the algorithms of Hopcroft and Kerr (1971).
Some of our bounds are proved to be optimal.  
}

\section{Introduction}

The time cost of an algorithm, sequential or parallel, depends not only on how many computational operations it executes but also on how much data it moves.  In fact, the cost of data movement, or \emph{communication}, is often much more expensive than the cost of computation.  Architectural trends predict that computation cost will continue to decrease exponentially faster than communication cost, leading to ever more algorithms that are dominated by the communication costs.  Thus, in order to minimize running times, algorithms should be designed with careful consideration of their communication costs.  To that end, we discuss asymptotic costs of algorithms in terms of both number of computations performed (flops in the case of numerical algorithms) and units of communication: \emph{words moved}.

For a sequential algorithm, we determine the communication cost incurred on a simple machine model which consists of two levels of memory hierarchy, as described in Section~\ref{sec:model}.  In many cases, na\"ive implementations of algorithms incur communication costs much higher than necessary; reformulating the algorithm to performing the same arithmetic in a different order can drastically decrease the communication costs and therefore the total running time.  In order to determine the possible improvements and identify whether an algorithm is optimal with respect to communication costs, one seeks communication lower bounds.

Hong and Kung \cite{HongKung81} were the first to prove communication lower bounds for matrix multiplication algorithms.  They show that on a two-level machine model, any algorithm which performs the $\Theta(n^3)$ flops of classical matrix multiplication must move at least $\Omega(n^3/\sqrt M)$ words between fast and slow memory, where $M$ is the number of words that can fit simultaneously in fast memory.  Irony, Toledo, and Tiskin \cite{IronyToledoTiskin04} generalized their classical matrix multiplication result to a distributed-memory parallel machine model using a \emph{geometric embedding} argument.  Ballard, Demmel, Holtz and Schwartz \cite{BallardDemmelHoltzSchwartz11a} showed this proof technique is applicable to a more general set of computations, including one-sided matrix factorizations such as LU, Cholesky, and QR and two-sided matrix factorizations which are used in eigenvalue and singular value computations, most of which perform $\Theta(n^3)$ computations in the dense matrix case.  Many of these bounds on \(\Theta(n^3)\) algorithms have been shown to be optimal.

However, the geometric embedding approach does not seem to apply to computations which do not map to a simple geometric computation space.  In the case of classical matrix multiplication and other $O(n^3)$ algorithms, the computation corresponds to a three-dimensional lattice.  In particular, the geometric embedding approach does not readily apply to Strassen's algorithm for matrix multiplication that requires $O(n^{\log_2 7})$ flops.  Instead, Ballard, Demmel, Holtz, and Schwartz \cite{BallardDemmelHoltzSchwartz11b} show that a different proof technique based on analysis of the expansion properties of the computational directed acyclic graph \emph{(CDAG)} can be used to obtain communication lower bounds for both sequential and parallel models for these algorithms.  The proof technique can also be used to bound how well the corresponding parallel algorithms can strongly-scale \cite{BallardDemmelHoltzLipshitzSchwartz12b}.  We use this same approach here to prove bounds on fast rectangular matrix multiplication algorithms, which introduce some extra technical challenges.

\subsection{Expansion and communication}
The CDAG of a recursive algorithm has a recursive structure,
and thus its expansion can be analyzed combinatorially (similarly
to what is done for expander graphs in \cite{Mihail89,AlonSchwartzShapira08,KouckyKabanetsKolokolova10}) or by spectral
analysis (in the spirit of what was done for the Zig-Zag expanders
\cite{ReingoldVadhanWigderson00}).  Analyzing the CDAG for communication cost bounds was first suggested by Hong and Kung
\cite{HongKung81}. They use the red-blue pebble game to obtain
tight lower bounds on the communication costs of many algorithms, including
classical $\Theta(n^3)$ matrix multiplication, matrix-vector multiplication, and
FFT. Their proof is obtained by considering dominator sets of the CDAG.

Other papers study connections between bounded space
computation and combinatorial expansion-related properties of the
corresponding CDAG (see e.g.,
\cite{Savage94,BilardiPreparata99,BilardiPietracaprinaD'Alberto00}
and references therein).  The study of expansion properties of a CDAG was also suggested as
one of the main motivations of Lev and Valiant \cite{LevValiant83}
in their work on superconcentrators and lower bounds on the arithmetic complexity of various problems.

\subsection{Fast rectangular matrix multiplication}
Following Strassen's algorithm for fast multiplication of square matrices \cite{Strassen69}, the arithmetic complexity of multiplying rectangular matrices has been extensively studied (see \cite{HopcroftKerr71,BiniCapovaniRomaniLotti79,Coppersmith82,LottiRomani83,HuangPan97,HuangPan98,Coppersmith97}  and further details in \cite{BurgisserClausenShokrollahi97}). When there is an algorithm for multiplying an \(m\times n\) matrix \(A\) with an \(n\times p\) matrix \(B\) to obtain an \(m\times p\) matrix \(C\) using only \(q\) scalar multiplications, we use the notation \(\mnpq mnpq\).\footnote{Recall that \(\mnpq mnpq\) implies that for all integers \(t\), \(\mnpq {m^t}{n^t}{p^t}{q^t}\) by recursion (tensor powering), and also that the arithmetic complexity of \(\langle m^t,n^t,p^t\rangle\) is \(O(q^t)\) regardless of the number of additions in \(\langle m,n,p\rangle\).}
The above studies try to minimize the number of multiplications \(q\) (as a function of $m,n,$ and $p$). A particular focus of interest is maximizing $\alpha$ so that
\(\mnpq{n}{n}{n^\alpha}{O(n^2\log n)}\) namely maximizing
the size of a rectangular matrix, so that it can be multiplied (from right) with a square matrix, in time which is only slightly more than what is needed to read the input.\footnote{Note that our approach may not apply to algorithms of the form \(\mnpq{n}{n}{n^\alpha}{O(n^2\log n)}\).  It only applies to algorithms that are a recursive application of a base-case algorithm.}
Recall that \(\langle m,n,p\rangle=\langle n,p,m\rangle=\langle p,m,n\rangle=\langle m,p,n \rangle = \langle p,n,m \rangle = \langle n,m,p \rangle\) for all \(m,n,p\) \cite{HopcroftMusinsky73}.

Rectangular matrix multiplication is used in many algorithms, for solving problems in linear algebra, in combinatorial optimization, and other areas. Utilizing fast algorithms for rectangular matrix multiplication has proved to be quite useful for improving the complexity
of solving many of those problems (a very partial list includes \cite{GalilPan89,Knight95,BelingMegiddo98,Zwick02,KratschSpinrad03,YusterZwick04,YusterZwick05,KaplanSharirVerbin06,KeZengHanPan08}).

\subsection{Communication model}
\label{sec:model}

We model communication costs on a sequential machine as follows.  Assume the machine has a fast memory of size \(M\) words and a slow memory of infinite size.  Further assume that computation can be performed only on data stored in the fast memory.  On a real computer, this model may have several interpretations and may be applied to anywhere in the memory hierarchy.  For example the slow memory might be the hard drive and the fast memory the DRAM; or the slow memory might be the DRAM and the fast memory the cache.

The goal is to minimize the number of words \(W\) transferred between fast and slow memory, which we call the communication cost of an algorithm.  Note that we minimize with respect to an algorithm, not with respect to a problem, and so the only optimization allowed is re-ordering the computation in a way that is consistent with the CDAG of the algorithm.  The sequential communication cost is closely related to communication costs in the various parallel models.  We discuss this relationship briefly in Section~\ref{sec:discussion}.

\subsection{The communication costs of rectangular matrix multiplication}\label{sec:intro:lb}
The communication costs lower bounds of rectangular matrix multiplication algorithms are determined by properties of the underlying CDAGs. 
Consider \(\langle m^t,n^t,p^t\rangle=q^t\) matrix multiplication that is generated from \(t\) tensor powers of \(\langle m,n,p\rangle=q\).  Denote the former by the {\em algorithm} and the latter by the {\em base case}, and consider their CDAGs.
They both consist of four parts: the encoding graphs of \(A\) and \(B\), the scalar multiplications, and the decoding graph of \(C\).  The encoding graphs correspond to computing linear combinations of entries of \(A\) or \(B\), and the decoding graph to computing linear combinations of the scalar products.  See Figure~\ref{fig:comp-graph} in Section~\ref{sec:exprect} for a diagram of the algorithm CDAG, and Figure~\ref{fig:binigraph} in Section~\ref{sec:examples} for an example of a base-case CDAG.
Let us state the communication cost lower bounds of the two main cases.
\begin{theorem} \label{thm:dec-con}
Let \(\mnpq{m^t}{n^t}{p^t}{q^t}\) be the algorithm obtained from a base case \(\mnpq mnpq\).  
If the decoding graph of the base case is connected, then the communication cost lower bound is 
\[W=\Omega\left(\frac{q^t}{M^{\log_{mp}q-1}}\right).\]
Further, in the case that \(n\leq m\) and \(n\leq p\) this bound is tight.
\end{theorem}

Note that in the case $m=n=p$, this result reproduces the lower bound for Strassen-like square matrix
multiplication algorithms in \cite{BallardDemmelHoltzSchwartz11b}.  In this case, for $\omega_0=\log_n q$, we obtain \(W=\Omega\left(\frac{(n^t)^{\omega_0}}{M^{\omega_0/2-1}}\right)\).

\begin{theorem} \label{thm:enc-con}
Let \(\mnpq{m^t}{n^t}{p^t}{q^t}\) be the algorithm obtained from a base case \(\mnpq mnpq\).  
If an encoding graph of the base case is connected and has no multiply-copied inputs\footnote{See Section~\ref{sec:preliminaries} for a formal definition.}, then
\[W=\Omega\left(\frac{q^t}{t^{\log_N q}M^{\log_{N}q-1}}\right),\]
where $N=mn$ or $N=np$ is the size of the input to the encoding graph.  Further, this bound is tight if \(N=\max\{mn,np,mp\}\), up to a factor of \(t^{\log_N q}\), which is a polylogarithmic factor in the input size.
\end{theorem}
We also treat the cases of disconnected encoding and decoding graphs and obtain similar bounds with restrictions on the fast memory size \(M\).  See Corollaries~\ref{thm:dec-discon} and~\ref{thm:enc-discon} in Section~\ref{sec:exprect}.

These theorems and corollaries apply in particular to the algorithms of Bini et al. \cite{BiniCapovaniRomaniLotti79} and Hopcroft and Kerr \cite{HopcroftKerr71}, which we detail in Section~\ref{sec:examples}.

\subsection{Paper organization} 
In Section~\ref{sec:preliminaries} we state some preliminary facts about the computational graph and edge expansion.  Section~\ref{sec:comp-exp} explains the connection between communication cost and edge expansion.  The proofs of the lower bound theorems stated in Section~\ref{sec:intro:lb}, as well as some extensions, appear in Section~\ref{sec:exprect}. In Section~\ref{sec:examples} we apply our new lower bounds to two example algorithms: Bini's algorithm and the Hopcroft-Kerr algorithm. Appendix~\ref{app:details} gives further details of Bini's algorithm and the Hopcroft-Kerr algorithm.

\section{Preliminaries}
\label{sec:preliminaries}

\subsection{The Computational Graph}
For a given algorithm, we consider the CDAG $G=(V,E)$, where there is a vertex for each arithmetic
operation {\em (AO)} performed, and for every input element. $G$
contains a directed edge $(u,v)$, if the output operand of the AO
corresponding to $u$ (or the input element corresponding to $u$),
is an input operand to the AO corresponding to $v$. The in-degree
of any vertex of $G$ is, therefore, at most 2 (as the arithmetic operations are binary).
The out-degree is, in general, unbounded, i.e., it may be a function of $|V|$.

\subsubsection{The relaxed computational graph.}
For a given recursive algorithm, the \emph{relaxed} computational graph is almost identical to the computational DAG with the following change: when a vertex corresponds to re-using data across recursive levels, we replace it with several connected ``copy vertices,'' each of which exists in one recursive level.  While the CDAG of a recursive algorithm may have vertices of degree that depend on $|V|$,    this relaxed CDAG has constant bounded degree. We use the relaxed graph to handle such cases in Section~\ref{sec:stretch}.

\subsubsection{Multiply-copied vertices.}
We say that a base-case encoding subgraph has {\em no multiply-copied vertices} if each input vertex appears at most once as an output vertex.  An output vertex \(v\) is {\em copied} from an input vertex if the in-degree of \(v\) is exactly one.  See, for example, Figure~\ref{fig:binigraph}.  The vertex \(a_{11}\) is copied to the third output of \(Enc_1A\) but is not copied to any other outputs.  Since all other inputs are also copied at most once, there are no multiply-copied vertices in Figure~\ref{fig:binigraph}.

This condition is necessary for the degree of the entire algorithm's encoding subgraph to be at most logarithmic in the size of the input.  We are not aware of any fast matrix multiplication algorithm that has multiply-copied vertices, although the recursive formulation of classical matrix multiplication does.

\subsection{Edge expansion}
The edge expansion $h(G)$ of a $d$-regular undirected
graph $G=(V,E)$ is:
\begin{equation*}
h(G) \equiv \min_{U \subseteq V, |U| \leq |V|/2} \frac{|E(U, V
\setminus U )|}{d \cdot |U|}
\end{equation*}
where $E(A,B)\equiv E_G(A,B)$ is the set of edges connecting the vertex
sets $A$ and $B$. We omit the subscript $G$ when
the context makes it clear.  Treating a CDAG as undirected simplifies the analysis and does not affect the asymptotic communication cost.  For many graphs, small sets expand more than larger sets. Let $h_s(G)$ denote the edge expansion for
sets of size at most $s$ in $G$:
\begin{equation*}
h_s(G) \equiv \min_{U \subseteq V, |U| \leq s} \frac{|E(U, V
\setminus U )|}{d \cdot |U|} ~.
\end{equation*}

Note that CDAGs are typically not regular.
If a graph $G=(V,E)$ is not regular but has a bounded maximal degree $d$,
then we can add ($<d$) loops to vertices of degree $<d$, obtaining
a regular graph $G'$. We use the convention that a loop adds 1 to the degree of a vertex.
Note that for any $S \subseteq V$,
we have $|E_{G}(S,V \setminus S)| = |E_{G'}(S,V \setminus S)|$, as
none of the added loops contributes to the edge expansion
of~$G'$.  

\subsection{Matching sequential algorithm}
In many cases, the communication cost lower bounds are matched by the na\"{i}ve recursive algorithm.  
The cost of the recursive algorithm applied to \(\langle m^t,n^t,p^t\rangle=q^t\), taking $N^*=\max\{mn,np,mp\}$ is
$$W(t)=\left\{\begin{array}{ll} q \cdot W(t-1) + \Theta\lt((N^*)^{t-1}\rt) & \textrm{if }(N^*)^t>M/3\\3(N^*)^t & \textrm{otherwise}\end{array}\right.,$$
since the algorithm does not communicate once the three matrices fit into fast memory.  The solution to this recurrence is given by
$$W = \Theta\lt(\frac{q^t}{M^{\log_{N^*}q-1}}\rt).$$

\section{Communication Cost and Edge Expansion}\label{sec:comp-exp}
In this section we recall the partition argument and how to combine it with edge expansion analysis to obtain communication cost lower bounds.  This follows our approach in \cite{BallardDemmelHoltzSchwartz11b,BallardDemmelHoltzLipshitzSchwartz12b}.  A similar partition argument previously appeared in \cite{HongKung81,IronyToledoTiskin04,BallardDemmelHoltzSchwartz11a}, where other techniques (geometric or combinatorial) are used to connect the number of flops to the amount of data in a segment.

\subsection{The partition argument}\label{sec:partition}
Let $M$ be the size of the fast memory. Let $O$ be any total
ordering of the vertices that respects the partial ordering of the
CDAG $G$. This total
ordering can be thought of as the actual order in which the
computations are performed. Let ${\cal P}$ be any partition of $V$ into
segments $S_1,S_2,...$, so that a segment $S_i \in {\cal P}$ is a subset of
the vertices that are contiguous in the total ordering $O$.

Let $R_S$ and $W_S$ be the set of read and write operands,
respectively. Namely, $R_S$ is the set
of vertices outside $S$ that have an edge going into $S$, and
$W_S$ is the set of vertices in $S$ that have an edge going
outside of $S$. Then the total communication costs due to reads of AOs in $S$ is
at least $|R_S|-M$, as at most $M$ of the needed $|R_S|$ operands
are already in fast memory when the execution of the segment's AOs
starts. Similarly, $S$ causes at least $|W_S|-M$ actual write
operations, as at most $M$ of the operands needed by other
segments are left in the fast memory when the execution of the
segment's AOs ends. The total communication cost is therefore bounded below
by

\begin{eqnarray}\label{eqn:RW} W &\geq& \min_{{\cal P}} \sum_{S \in {\cal P}}
\lt( |R_S| + |W_S| - 2M \rt)~.
\end{eqnarray}

\subsection{Edge expansion and communication cost}

Consider a segment $S$ and its read and write operands $R_S$ and
$W_S$. 
\begin{proposition}\label{clm:RW}
If the graph $G$ containing $S$ has $h_s(G)$ edge expansion\footnote{For many algorithms, the edge expansion $h(G)$ deteriorates with $|G|$, whereas $h_s(G)$ is constant with respect to $|G|$, which allows for better communication lower bounds.} for sets of size \(s=|S|\), maximum
(constant) degree $d$, and at least $2|S|$ vertices, then
$|R_S|+|W_S| \geq \frac12 \cdot h_s(G) \cdot |S|$~.
\end{proposition}
\begin{proof}
We have $|E(S, V \setminus S)| \geq h_s(G) \cdot d \cdot |S|$.
Either (at least) half of the edges $E(S, V \setminus S)$ touch
$R_S$ or half of them touch $W_S$. As every vertex is of degree
$d$, we have $|R_S|+|W_S| \geq \max\{|R_S|,|W_S|\} \geq \frac 1d
\cdot \frac12 \cdot |E(S, V \setminus S)| \geq h_s(G) \cdot |S| /2$.
\end{proof}

\noindent
Combining this with (\ref{eqn:RW}) and choosing to partition
$V$ into $|V|/s$ segments of equal size $s$, we obtain:
$
W \geq  \max_{s}\frac{|V|}{s} \cdot \lt( \frac{h_s(G) \cdot
s}{2} - 2M \rt)
$.
Choosing the minimal $s$ so that
\begin{eqnarray}\label{eqn:IO-general-expansion-condition}
\frac{h_s(G) \cdot s}{2} \geq 3M
\end{eqnarray} we obtain
\begin{eqnarray}\label{eqn:IO-general-expansion}
W &\geq&  \frac{|V|}{s} \cdot M~.
\end{eqnarray}

In some cases, as in fast square and rectangular matrix multiplication, the computational graph $G$ does not fit this
analysis: it may not be regular, it may have vertices of unbounded
degree, or its edge expansion may be hard to analyze. In such
cases, we may then consider some subgraph $G'$ of $G$ instead to
obtain a lower bound on the communication cost.  The natural subgraph to select in fast (square and rectangular) matrix multiplication algorithms is the decoding graph or one of the two encoding graphs.

\section{Expansion Properties of Fast Rectangular Matrix Multiplication Algorithms}\label{sec:exprect}

There are several technical challenges that we deal with in the rectangular case, on top of the analysis in \cite{BallardDemmelHoltzSchwartz11b} (where we deal with the difference between addition and multiplication vertices in the recursive construction of the CDAG).  These additional challenges arise from the differences between the CDAG of rectangular algorithms, such as Bini's algorithm and the Hopcroft-Kerr algorithm on the one hand, and of Strassen's algorithm on the other hand.  The three subgraphs, two encoding and one decoding, are of the same size in Strassen's and of unequal size in rectangular algorithms.  The largest expansion guarantee is given by the subgraph corresponding to the largest of the three matrices.  One consequence is that it is necessary to consider the case of unbounded degree vertices that may appear in the encoding subgraphs.  Additionally, in some cases the encoding or decoding graphs consist of several disconnected components.

\subsection{The computational graph for $\mnpq{m^t}{n^t}{p^t}{q^t}$}
Consider the computational graph $H_t$ associated with multiplying a matrix \(A\) of dimension \(m^t\times n^t\) by a matrix \(B\) of dimension \(n^t\times p^t\).
Denote by
$Enc_{t}A$ the part of $H_{t}$ that corresponds to the encoding
of matrix $A$. Similarly, $Enc_{t}B$, and $Dec_{t}C$
correspond to the parts of $H_{t}$ that compute the encoding
of $B$ and the decoding of $C$, respectively  (see Figure~\ref{fig:comp-graph}).

\begin{figure}[t]
\begin{center}
\includegraphics[width=3in]{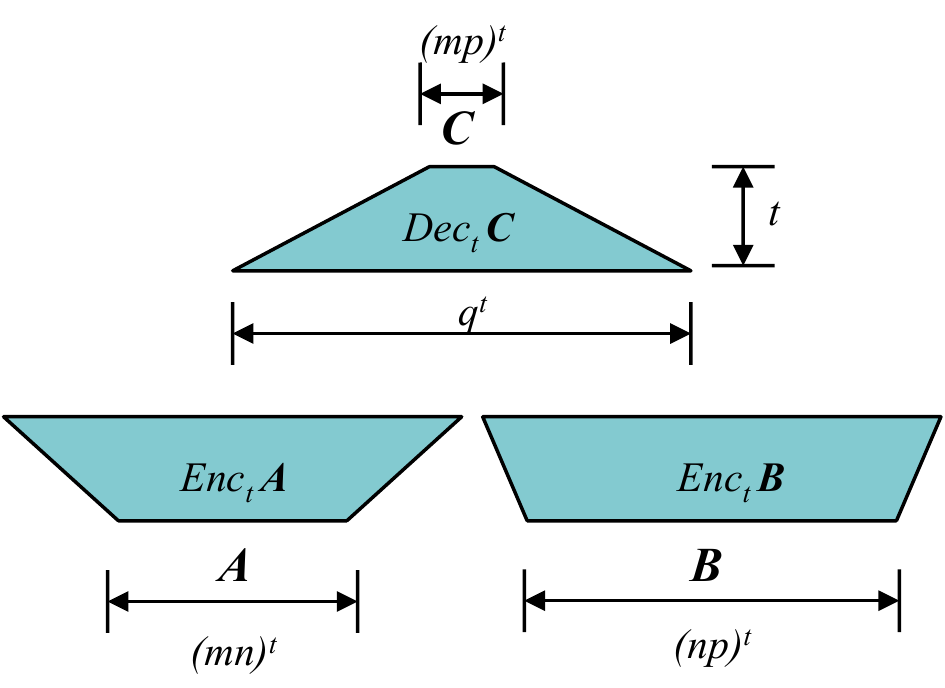}
\caption{Computational graph for $\langle m^t,n^t,p^t \rangle = q^t$ rectangular matrix multiplication generated from $t$ recursive levels with base graph given by $\langle m,n,p \rangle = q$.  In this figure $m<p<n$.}
\label{fig:comp-graph}
\end{center}
\end{figure}

\subsubsection{A top-down construction of the computational graph.}\label{sec:recconst}
We next construct the computational graph $H_{i+1}$ by constructing $Dec_{i+1}C$ from $Dec_i C$ and $Dec_1 C$
and similarly constructing $Enc_{i+1}A$ and $Enc_{i+1}B$, then composing the three parts together.
\begin{enumerate}
\item
Duplicate $Dec_{1}C$ $q^i$ times.
\item
Duplicate $Dec_{i}C$ $mp$ times.
\item
Identify the $mp \cdot q^i$ output vertices of the copies of
$Dec_{1}C$ with the $mp \cdot q^i$ input vertices of the copies of $Dec_{i}C$:
\begin{itemize}
\item
Recall that each $Dec_{1}C$ has \(mp\) output vertices.
\item
The first output vertex of the $q^i$
$Dec_{1}C$ graphs are identified with the $q^i$ input vertices of the first copy of $Dec_{i}C$.
\item
The second output vertex of the $q^i$
$Dec_{1}C$ graphs are identified with the $q^i$ input vertices of the second copy of $Dec_{i}C$.
And so on.
\item
We make sure that the $j$th input vertex of a copy of $Dec_{i}C$ is identified
with an output vertex of the $j$th copy of $Dec_{1}C$.
\end{itemize}
\item
We similarly obtain $Enc_{i+1}A$ from $Enc_{i}A$ and $Enc_{1}A$,
\item
and $Enc_{i+1}B$ from $Enc_{i}B$ and $Enc_{1}B$.
\item
For every $i$, $H_{i}$ is obtained by connecting edges from the $j$th output vertices of $Enc_{i}A$ and $Enc_{i}B$
to the $j$th input vertex of $Dec_{i}C$.
\end{enumerate}
This completes the construction. Let us note some properties of this graphs.

As all out-degrees are at most \(mp\) and all in degree are at most 2 we have:
\begin{proposition}\label{fct:constant-degree}
All vertices of $Dec_{t}C$ are of degree at most $mp+2$, as long as \(n>1\) (that is, as long as the base case is not an outer product).
\end{proposition}
\begin{proof}
If the set of input vertices of $Dec_1C$ and the set of its
output vertices are disjoint, then the proposition follows..
Assume (towards contradiction) that the base graph $Dec_1C$ has
an input vertex which is also an output vertex. An output vertex
represents the inner product of two $n$-vectors, i.e., the corresponding row-vector of $A$ and column vector of $B$. The corresponding bilinear polynomial is irreducible.
This is a
contradiction, since \(n>1\)
an input
vertex represents the multiplication of a (weighted) sum of elements of
$A$ with a (weighted) sum of elements of $B$.
\end{proof}

Note, however, that $Enc_1A$ and $Enc_1B$ may have vertices which are both inputs and outputs, therefore $Enc_{t}A$ and $Enc_{t}B$ may have vertices of out-degree which is a function of $t$.
In \cite{BallardDemmelHoltzSchwartz11b,BallardDemmelHoltzLipshitzSchwartz12b}, it was enough to analyze \(Dec_tC\) and  lose only a constant factor in the lower bound.  However in several rectangular matrix multiplication algorithms, it is necessary to consider the encoding graphs as well, since they may provide a better expansion than the decoding graph.

\begin{lemma}\label{lem:Dec-expansion}
If \(Dec_1C\) is connected, then the edge expansion of $Dec_{t}C$ is $$h(Dec_{t}C)=\Omega\lt(\lt(\frac{mp}{q}\rt)^{t}\rt).$$
\end{lemma}

\begin{proof}
The proof follows that of Lemma~4.9 in \cite{BallardDemmelHoltzSchwartz11b} adapting the corresponding parameters.  We provide it here for completeness.
Let $G_t=(V,E)$ be $ Dec_tC$, and let $S \subseteq V, |S| \leq
|V|/2$. We next show that $|E(S,V \setminus S)| \geq c \cdot d
\cdot |S| \cdot \lt(\frac{mp}q \rt) ^t $, where $c$ is some universal
constant, and $d$ is the constant degree of $Dec_tC$ (after
adding loops  to make it regular).

The proof works as follows. Recall that $G_t$ is a layered graph (with layers corresponding to recursion steps),
so all edges (excluding loops)
connect between consecutive levels of vertices. We argue (in Proposition \ref{clm:S-homogeneity}) that each level of
$G_t$ contains about the same fraction of $S$ vertices, or else we have
many edges leaving $S$. We also observe (in Fact \ref{fct:heterogeneity}) that such
homogeneity (of a fraction of $S$ vertices) does not hold between
distinct parts of the lowest level, or, again, we have many
edges leaving $S$. We then show that the homogeneity between levels,
combined with the heterogeneity of the lowest level,
guarantees that there are many edges leaving $S$.

Let $l_i$ be the $i$th level of vertices of $G_t$, so $(mp)^t = |l_1|
< |l_2| < \cdots < |l_i| = (mp)^{t-i+1}q^{i-1} < \cdots < |l_{t+1}| = q^t$.
Let $S_i \equiv S \cap l_i$. Let $\sigma = \frac{|S|}{|V|}$ be the fractional size
of $S$ and $\sigma_i = \frac{|S_i| }{|l_i|}$ be the fractional size of $S$ at level $i$.  Let $\delta_i=\sigma_i-\sigma_{i+1}$.
Due to averaging, we observe the following:
\begin{fact} There
exist $i$ and $i'$ such that $\sigma_i \leq \sigma \leq \sigma_{i'}$.
\end{fact}

\begin{fact}\label{fct:l1}
\begin{eqnarray*}
|V| &= &\sum _{i =1}^{t+1} |l_{i}|
=
 \sum_{i =1}^{t+1} |l_{t+1}| \cdot \lt(\frac{mp}{q}\rt)^i\\
&=&  |l_{t+1}| \cdot \lt(1 - \lt(\frac{mp}{q}\rt)^{t+2}\rt) \cdot \frac q{q-mp}\\
&=&  \left(\frac{mp}q \right)^t\cdot |l_{1}| \cdot \lt(1 - \lt(\frac{mp}{q}\rt)^{t+2}\rt) \cdot \frac q{q-mp}.
\end{eqnarray*}
so
$ \frac{q-mp}q \leq \frac{|l_{t+1}|}{|V|} \leq \frac{q-mp}q \cdot \frac{1}{1- \lt(\frac{mp}{q}\rt)^{t+2}  }
$, and
$
\frac{q-mp}q \cdot \lt(\frac{mp}{q}\rt)^{t}
\leq \frac{|l_1|}{|V|} \leq \frac{q-mp}q \cdot \lt(\frac{mp}{q}\rt)^{t} \cdot \frac{1}{1- \lt(\frac{mp}{q}\rt)^{t+2}  }.
$
\end{fact}

\begin{proposition}\label{clm:delta-cost}
There exists $c' = c'(G_1)$ so that $|E(S, V \setminus S) \cap E(l_i,l_{i+1})| \geq c'\cdot d \cdot|\delta_i| \cdot
|l_i| $.
\end{proposition}
\begin{pfof}{Proposition~\ref{clm:delta-cost}}
Let $G'$ be a $G_1$ component connecting $l_i$ with $l_{i+1}$ (so it has \(mp\) vertices in $l_i$ and \(q\) in $l_{i+1}$).
$G'$ has no edges in $E(S,V \setminus
S)$ if all or none of its vertices are in $S$. Otherwise, as $G'$ is connected, it
contributes at least one edge to $E(S,V \setminus S)$.
The number of such $G_1$ components with all their vertices in $S$
is at most $\min\{\sigma_i,\sigma_{i+1}\}\cdot \frac{|l_i|}{mp}$. Therefore, there are at least
$|\sigma_{i}-\sigma_{i+1}|\cdot \frac{|l_i|}{mp}$  $G_1$ components with at least one vertex in $S$ and one vertex that is not.
\end{pfof}

\begin{proposition}[Homogeneity between levels]\label{clm:S-homogeneity}
If there exists $i$ so that $\frac{|\sigma-\sigma_i|}{\sigma}\geq \frac1{10}$,  then
$$|E(S,V \setminus S)| \geq c \cdot d \cdot |S|\cdot \lt( \frac{mp}q\rt)^t $$
where $c>0$ is some constant depending on $G_1$ only.
\end{proposition}
\begin{pfof}{Proposition~\ref{clm:S-homogeneity}}
Assume that there exists $j$ so that $\frac{|\sigma-\sigma_j|}{\sigma}\geq \frac1{10}$. 
By Proposition \ref{clm:delta-cost}, we have
\begin{eqnarray*}
|E(S,V \setminus S)|
&\geq& \sum _{i \in [t]} |E(S, V \setminus S) \cap E(l_i,l_{i+1})|\\
&\geq& \sum _{i \in [t]} c'\cdot d\cdot|\delta_i| \cdot |l_i|\\
&\geq& c'\cdot d\cdot |l_1| \sum _{i \in [t]}   |\delta_i|\\
&\geq& c'\cdot d\cdot |l_1|
 \cdot \lt(\max_{i \in [t+1]} \sigma_i
- \min_{i \in [t+1]} \sigma_i\rt) .
\end{eqnarray*}
By the initial assumption, there exists $j$ so that $\frac{|\sigma-\sigma_j|}{\sigma}\geq \frac1{10}$, therefore $\max_i \sigma_i - \min_i \sigma_i \geq \frac{\sigma}{10}$, then
\begin{align*}
|E(S,V \setminus S)|
&\geq  c'\cdot d\cdot |l_1| \cdot \frac{\sigma}{10}\\
\intertext{By Fact \ref{fct:l1}, $|l_1| \geq  \frac{q-mp}q \cdot \lt ( \frac{mp}{q} \rt)^t \cdot |V| $, }
&\geq c'\cdot d \cdot \frac{q-mp}q \cdot \lt( \frac{mp}{q} \rt)^t \cdot |V| \cdot \frac{\sigma}{10}\\
\intertext{As $|S| = \sigma \cdot |V|$,}
&\geq c \cdot d\cdot |S| \cdot \lt( \frac{mp}q\rt)^t
\end{align*}
for any $c \leq \frac{c'}{10}\cdot \frac{q-mp}q$.
\end{pfof}

Let $T_t$ be a tree corresponding to the recursive construction of $G_t$ in the following way: 
$T_t$ is a tree of height $t+1$, where each internal node has \(mp\) children.
The root $r$ of $T_t$ corresponds to $l_{t+1}$ (the largest level of $G_t$).
The \(mp\) children of $r$ correspond to the largest levels of the \(mp\) graphs that one can obtain by
removing the level of vertices $l_{t+1}$ from $G_t$. And so on.
For every node $u$ of $T_t$, denote by $V_u$ the set of vertices in $G_t$ corresponding to $u$.
We thus have $|V_r|=q^t$ where $r$ is the root of $T_t$,
$|V_u| = q^{t-1}$ for each node $u$ that is a child of $r$;
and in general we have $(mp)^{i}$ tree nodes $u$ corresponding to a set of size $|V_u| = q^{t-i+1}$.
Each leaf $l$ corresponds to a set of size $1$.

For a tree node $u$, let us define $\rho_u = \frac{|S \cap V_u|}{|V_u|}$ to be the fraction of $S$ nodes in $V_u$,
and $\delta_u = |\rho_u - \rho_{p(u)}|$, where $p(u)$ is the parent of $u$ (for the root $r$ we let $p(r)=r$).
We let $t_i$ be the $i$th level of $T_t$, counting from the bottom, so $t_{t+1}$ is the root and $t_{1}$ are the leaves.

\begin{fact}\label{fct:heterogeneity}
As $V_r=l_{t+1}$ we have $\rho_r = \sigma_{t+1}$. For a tree leaf $u \in t_1$, we have $|V_u|=1$. Therefore $\rho_u \in \{0,1\}$.
The number of vertices $u$ in $t_1$ with $\rho_u=1$ is $\sigma_1 \cdot |l_1|$.
\end{fact}

\begin{proposition}\label{clm:delta-tree-cost}
Let $u_0$ be an internal tree node, and let $u_1,u_2,\ldots,u_{mp}$ be its \(mp\) children. Then
$$\sum_i|E(S, V \setminus S) \cap E(V_{u_i},V_{u_0})| \geq c''\cdot d \cdot \sum_i  |\rho_{u_i}-\rho_{u_0}| \cdot |V_{u_i}|$$
where $c''=c''(G_1)$.
\end{proposition}
\begin{pfof}{Proposition~\ref{clm:delta-tree-cost}}
The proof follows that of Proposition \ref{clm:delta-cost}.
Let $G'$ be a $G_1$ component connecting $V_{u_0}$ with $\bigcup_{i \in [mp]}V_{u_i}$ (so it has \(q\) vertices in $V_{u_0}$
and one in each of $V_{u_1}$,$V_{u_2}$,\ldots,$V_{u_{mp}}$).
$G'$ has no edges in $E(S,V \setminus
S)$ if all or none of its vertices are in $S$. Otherwise, as $G'$ is connected, it
contributes at least one edge to $E(S,V \setminus S)$.
The number of $G_1$ components with all their vertices in $S$
is at most $\min\{\rho_{u_0},\rho_{u_1},\rho_{u_2},\dots,\rho_{u_{mp}}\}\cdot \frac{|V_{u_1}|}{mp}$. Therefore, there are at least
$\max_{i \in [mp]}\{|\rho_{u_0}-\rho_{u_i}|\}\cdot \frac{|V_{u_1}|}{mp} \geq \frac{1}{(mp)^2} \cdot \sum_{i\in [mp]} |\rho_{u_i}-\rho_{u_0}| \cdot |V_{u_i}|$ \; $G_1$  components with at least one vertex in $S$ and one vertex that is not.
\end{pfof}

\begin{align*}
\intertext{We have}
\nonumber
|E(S,V \setminus S)|
&= \sum _{u \in T_t} |E(S, V \setminus S) \cap E(V_u,V_{p(u)})|\\
\nonumber
\intertext{By Proposition \ref{clm:delta-tree-cost}, this is}
\nonumber
&\geq  \sum _{u \in T_t} c''\cdot d\cdot  |\rho_{u}-\rho_{p(u)}|  \cdot |V_u|\\
&= c''\cdot d\cdot \sum_{i \in [t]} \sum _{u \in t_i} |\rho_{u}-\rho_{p(u)}|  \cdot q^{i-1}\\
\nonumber
&\geq c''\cdot d\cdot \sum_{i \in [t]} \sum _{u \in t_i}  |\rho_u-\rho_{p(u)}| \cdot (mp)^{i-1} \\
\nonumber
\intertext{As each internal node has \(mp\) children, this is}
&= c''\cdot d\cdot \sum _{v \in t_1} \sum_{u\in v \sim r}  |\rho_u-\rho_{p(u)}| \\
\nonumber
\intertext{where $v \sim r$ is the path from $v$ to the root $r$. By the triangle inequality for the function~$|~\cdot~|$}
&\geq c''\cdot d\cdot \sum _{v \in t_1}  |\rho_u-\rho_r|\\
\nonumber
\intertext{By Fact \ref{fct:heterogeneity}, }
&\geq  c''\cdot d\cdot  |l_1| \cdot ((1-\sigma_1) \cdot \rho_{r} + \sigma_1 \cdot (1-\rho_{r}) )\\
\intertext{By Proposition \ref{clm:S-homogeneity}, w.l.o.g.,
$|\sigma_{t+1} - \sigma|/\sigma \leq \frac{1}{10}$ and
$|\sigma_1 - \sigma|/\sigma \leq \frac{1}{10}$. As $\rho_r = \sigma_{t+1}$,}
& \geq \frac34 \cdot c''\cdot d\cdot  |l_1| \cdot \sigma  \\
\intertext{and by Fact \ref{fct:l1},}
&\geq c\cdot d\cdot  |S| \cdot \left(\frac{mp}q \right)^{t}\\
\end{align*}
for any $c \leq \frac34 \cdot c''$.
\end{proof}

Using Lemma 2.1 of \cite{BallardDemmelHoltzSchwartz11b} (decomposition into edge disjoint small subgraphs) we deduce that for sufficiently large \(t\),
$$h_s(Dec_tC)=\Omega\lt(\lt(\frac{mp}{q}\rt)^{\log_q s}\rt).$$
Thus there exists a constant \(c\) such that for \(s=cM^{\log_{mp}q}\), 
$s \cdot h_{s}(Dec_{t}C) \geq 3M$.
Plugging this into inequality~(\ref{eqn:IO-general-expansion}) we obtain Theorem~\ref{thm:dec-con}.

\subsection{Stretching a segment}
\label{sec:stretch}

We next consider the case where all vertices have a degree bounded by \(O(t)\).  We analyze the edge expansion of the relaxed computational graph,\footnote{See Section~\ref{sec:preliminaries} for a formal definition.} which corresponds to the same set of computations but has a constant degree bound.  We then show that an augmented partition argument (similar to that in Section~\ref{sec:partition}) results in a communication cost lower bound which is optimal up to at most a polylogarithmic factor.

Since a relaxed encoding graph has a constant degree bound we can analyze the expansion of the \(Enc_tA\) and \(Enc_tB\) parts of the computational graph by exactly the same technique used for \(Dec_tC\) above.  Plugging in the corresponding parameters, we thus obtain:
\begin{lemma}\label{lem:stretch}
Let \(G'_t\) be the relaxed computational graph of computing \(\mnpq{m^t}{n^t}{p^t}{q^t}\) based on $\langle m,n,p \rangle = q$.  Let \(Enc'_tA\) and \(Enc'_tB\) be the subgraphs corresponding to the encoding of \(A\) and \(B\) in \(G'_t\).  Then
  \[h_s(Enc'_tA)=\Omega\left(\left(\frac{mn}{q}\right)^{\log_q s}\right) 
  \quad \text{ and } \quad
  h_s(Enc'_tB)=\Omega\left(\left(\frac{np}{q}\right)^{\log_q s}\right).\]
\end{lemma}

Consider a CDAG \(G\) with maximum degree \(O(t)\) and its corresponding relaxed CDAG \(G'\) of constant degree.  Given the expansion of \(G'\) we would like to deduce the communication cost incurred by computing \(G\).  To this end we need amended versions of inequalities~(\ref{eqn:IO-general-expansion-condition}) and~(\ref{eqn:IO-general-expansion}); since by transforming $G'$ back to $G$ \(|R_s|+|W_s|\) may contract by a factor of \(O(t)\), we need to compensate for that by increasing the segment size \(s\).  To be precise, we want 
\(\frac{|R_s|+|W_s|}{ct}-2M=M.\)
Following inequality~(\ref{eqn:IO-general-expansion-condition}), we thus choose the minimal \(s\) such that \(h_s(Enc_tA)\cdot s\geq c'tM\), where \(c'\) is some universal constant.  By inequality~(\ref{eqn:IO-general-expansion}) and~Lemma~\ref{lem:stretch}, 
\(\left(\frac{mn}{q}\right)^{\log_q s}\cdot s=\Theta(tM),\)
so 
\[W=\Omega\left(\frac{q^t}{(tM)^{\log_{mn}q}}M\right)\]
and Theorem~\ref{thm:enc-con} follows.

\subsection{Disconnected encoding or decoding graphs}
The CDAG of any fast (rectangular or square) matrix multiplication algorithm must be connected, due to the dependencies of the output entries on the input entries.  The encoding and decoding graphs, however, are not always connected (see e.g., Bini's algorithm, in Section~\ref{sec:examples:bini} and Appendix~\ref{app:details}).  Consider a case where each connected components of \(Dec_tC\) is small enough to fit into the fast memory.  Then our proof technique cannot provide a nontrivial lower bound.  Even if a connected component is larger than \(M\), but has \(\leq M\) inputs and \(\leq M\) outputs, the partition into segments approach provides no communication cost lower bound (see inequality~(\ref{eqn:RW}) and its proof).  In the case that the inputs of an encoding graph or the output of the decoding graph do not fit into fast memory, and the disconnected components all have the same number of input and output vertices, the lower bound technique still applies.  Formally,

\begin{corollary} \label{thm:dec-discon}
If the base-case decoding graph is disconnected and consists of \(X\) connected components of equal input and output size, then
\(W=\Omega\left(\frac{q^t}{M^{\log_{mp/X}(q/X)-1}}\right).\)
\end{corollary}

\begin{proof}
Since \(Dec_tC\) is disconnected \(h(Dec_tC)=0\). However it consists of \(X^t\) connected components, each of which has nonzero expansion, therefore the entire graph does have expansion for small sets.  Each connected component is recursively constructed from a base graph with \(q/X\) inputs and \(mp/X\) outputs.  By Lemma~\ref{lem:Dec-expansion}, each connected component \(CC_t\) of \(Dec_tC\) has expansion
\[h(CC_t)=\Omega\left(\left(\frac{mp}{q}\right)^t\right).\]
In order to apply  Lemma 2.1 of \cite{BallardDemmelHoltzSchwartz11b} (decomposition into edge disjoint small subgraphs), we decompose \(Dec_tC\) into connected components of size \(s\), where \(s\) needs to satisfy two conditions.  First, \(s\) must be smaller than the size of the connected components of \(Dec_tC\) (otherwise we cannot claim any expansion), namely
\[s=O\left(\left(\frac{q}{X}\right)^t\right).\]
Second, \(s\) must be large enough so that the output of one component does not fit into fast memory (otherwise the expansion guarantee does not translate into a communication lower bound):
\[\left(\frac{mp}{X}\right)^{k}=\Omega(M),\]
where \(k=\log_{q/X}s\) is the number of recursive steps inside one component.
We then deduce that
\[h_s(Dec_tC)=\Omega\left(\left(\frac{mp}{q}\right)^{\log_{q/X}s}\right).\]
Thus there exists a constant \(c\) such that for \(s=cM^{\log_{mp/X}(q/X)}\), 
$s \cdot h_{s}(Dec_{t}C) \geq 3M$.
Plugging this into inequality~(\ref{eqn:IO-general-expansion}) we obtain Corollary~\ref{thm:dec-discon}.  Note that in the case that \(M=\Omega\lt(\left(\frac{mp}{X}\right)^t\rt)\), the argument above does not apply, but the result still holds because it is weaker than the trivial bound that the entire output must be written: \(W=\Omega\left((mp)^t\right)\).
\end{proof}

\begin{corollary} \label{thm:enc-discon}
If a base-case encoding graph is disconnected and consists of \(X\) connected components of equal input and output size, has \(N\) inputs, where \(N=mn\) or \(N=np\), and has no multiply-copied inputs, then
\(W=\Omega\left(\frac{q^t}{t^{\log_{N/x}(q/X)}M^{\log_{N/X}(q/X)-1}}\right).\)
\end{corollary}

\begin{proof}
Let \(G'_t\) be the relaxed computational graph of computing \(\mnpq{m^t}{n^t}{p^t}{q^t}\) based on $\langle m,n,p \rangle = q$.  Let \(Enc'_t\) be the subgraph corresponding to the encoding of \(A\) or \(B\) in \(G'_t\), and \(N\) be \(mn\) (for the encoding of \(A\)) or \(np\) (for the encoding of \(B\)).  Then by the same argument as above,
  \[h_s(Enc'_t)=\Omega\left(\left(\frac{N}{q}\right)^{\log_{q/X} s}\right).\]

Since by transforming $G'$ back to $G$ the sum \(|R_s|+|W_s|\) may contract by a factor of \(O(t)\) (recall Section~\ref{sec:stretch}), we need to compensate for that by increasing the segment size \(s\).  Thus the above only holds for
\[\left(\frac{N}{X}\right)^{k}=\Omega(Mt),\]
where \(k=\log_{q/X}s\).
It follows that there exists a constant \(c\) such that for \(s=c(tM)^{\log_{mp/X}(q/X)}\), 
$s \cdot h_{s}(Enc'_{t}) \geq 3tM$.
Plugging this into inequality~(\ref{eqn:IO-general-expansion}) we obtain Corollary~\ref{thm:enc-discon}.
Note that in the case that \(M=\Omega\lt(\left(\frac{N}{X}\right)^t\rt)\), the argument above does not apply, but the result still holds because it is weaker than the trivial bound that the entire input must be read: \(W=\Omega\left(N^t\right)\).
\end{proof}

\section{The Communication Costs of Some Rectangular Matrix Multiplication Algorithms}
\label{sec:examples}

In this section we apply our main results to get new lower bounds for rectangular algorithms based on Bini's algorithm \cite{BiniCapovaniRomaniLotti79} and the Hopcroft-Kerr algorithm \cite{HopcroftKerr71}.  All rectangular algorithms yield a square algorithm.  In the case of Bini the exponent is \(\omega_0\approx 2.779\), slightly better than Strassen's algorithm (\(\omega_0\approx2.807\)), and in the case of Hopcroft-Kerr the exponent is \(\omega_0\approx2.811\), slightly worse than Strassen's algorithm.  These algorithms are stated explicitly, which is not true of most of the recent results that significantly improve \(\omega_0\).  See Table~\ref{table:examples} for an enumeration of several algorithms based on \cite{BiniCapovaniRomaniLotti79,HopcroftKerr71} and their lower bounds.

\subsection{Bini's algorithm}\label{sec:examples:bini}
Bini et~al.~\cite{BiniCapovaniRomaniLotti79} obtained the first approximate matrix multiplication algorithm.  They introduce a parameter \(\lambda\) into the computation and give an algorithm that computes matrix multiplication up to terms of order \(\lambda\).  It was later shown how to convert such approximate algorithms into exact algorithms without changing the asymptotic arithmetic complexity, ignoring logarithmic factors \cite{Bini80}.\footnote{We treat here the original, approximate algorithm, not any of the exact algorithms that can be derived from it.}

Bini et al. show how to compute \(2\times2\times2\) matrix multiplication approximately where one of the off-diagonal entries of an input matrix is zero using 5 scalar multiplications.  This can be used twice to give an algorithm for \(\langle3,2,2\rangle=10\) matrix multiplication.  Notably this algorithm has disconnected \(Enc_1 A\) (see Figure~\ref{fig:binigraph}).
\begin{figure}
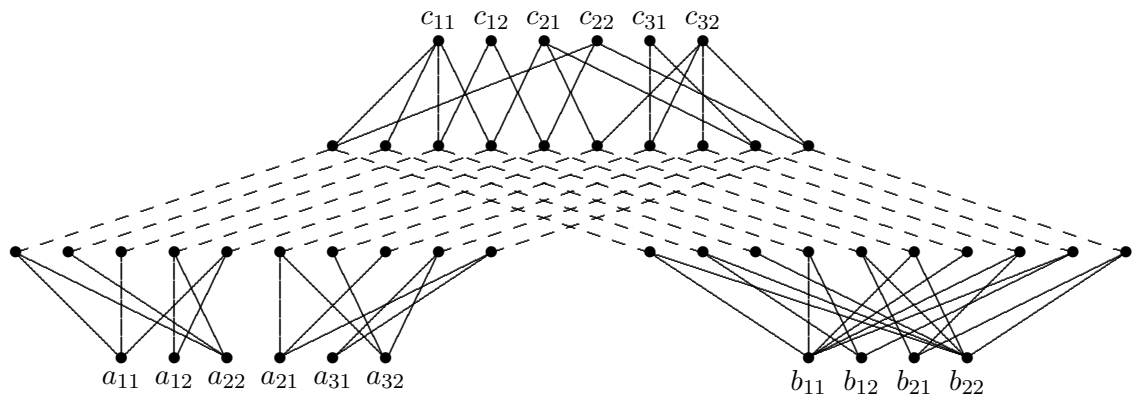

$$
\begindc{\undigraph}[20]
\obj(6,4)[q1]{}
\obj(7,4)[q2]{}
\obj(8,4)[q3]{}
\obj(9,4)[q4]{}
\obj(10,4)[q5]{}
\obj(11,4)[q6]{}
\obj(12,4)[q7]{}
\obj(13,4)[q8]{}
\obj(14,4)[q9]{}
\obj(15,4)[q10]{}
\obj(8,6)[c1]{$c_{11}$}
\obj(10,6)[c2]{$c_{21}$}
\obj(12,6)[c3]{$c_{31}$}
\obj(9,6)[c4]{$c_{12}$}
\obj(11,6)[c5]{$c_{22}$}
\obj(13,6)[c6]{$c_{32}$}
\mor{q1}{c1}{}
\mor{q2}{c1}{}
\mor{q3}{c1}{}
\mor{q4}{c1}{}
\mor{q4}{c2}{}
\mor{q6}{c2}{}
\mor{q9}{c2}{}
\mor{q7}{c3}{}
\mor{q9}{c3}{}
\mor{q3}{c4}{}
\mor{q5}{c4}{}
\mor{q1}{c5}{}
\mor{q5}{c5}{}
\mor{q10}{c5}{}
\mor{q6}{c6}{}
\mor{q7}{c6}{}
\mor{q8}{c6}{}
\mor{q10}{c6}{}
\obj(0,2)[t1]{}
\obj(1,2)[t2]{}
\obj(2,2)[t3]{}
\obj(3,2)[t4]{}
\obj(4,2)[t5]{}
\obj(5,2)[t6]{}
\obj(6,2)[t7]{}
\obj(7,2)[t8]{}
\obj(8,2)[t9]{}
\obj(9,2)[t10]{}
\obj(2,0)[a1]{$a_{11}$}[5]
\obj(3,0)[a2]{$a_{12}$}[5]
\obj(4,0)[a4]{$a_{22}$}[5]
\obj(5,0)[a3]{$a_{21}$}[5]
\obj(6,0)[a5]{$a_{31}$}[5]
\obj(7,0)[a6]{$a_{32}$}[5]
\mor{t1}{a1}{}
\mor{t3}{a1}{}
\mor{t5}{a1}{}
\mor{t6}{a3}{}
\mor{t8}{a3}{}
\mor{t10}{a3}{}
\mor{t1}{a4}{}
\mor{t2}{a4}{}
\mor{t4}{a4}{}
\mor{t4}{a2}{}
\mor{t5}{a2}{}
\mor{t9}{a5}{}
\mor{t10}{a5}{}
\mor{t6}{a6}{}
\mor{t7}{a6}{}
\mor{t9}{a6}{}
\obj(12,2)[s1]{}
\obj(13,2)[s2]{}
\obj(14,2)[s3]{}
\obj(15,2)[s4]{}
\obj(16,2)[s5]{}
\obj(17,2)[s6]{}
\obj(18,2)[s7]{}
\obj(19,2)[s8]{}
\obj(20,2)[s9]{}
\obj(21,2)[s10]{}
\obj(15,0)[b1]{$b_{11}$}[5]
\obj(17,0)[b2]{$b_{21}$}[5]
\obj(16,0)[b3]{$b_{12}$}[5]
\obj(18,0)[b4]{$b_{22}$}[5]
\mor{s1}{b1}{}
\mor{s4}{b1}{}
\mor{s6}{b1}{}
\mor{s7}{b1}{}
\mor{s8}{b1}{}
\mor{s9}{b1}{}
\mor{s5}{b2}{}
\mor{s8}{b2}{}
\mor{s10}{b2}{}
\mor{s2}{b3}{}
\mor{s4}{b3}{}
\mor{s9}{b3}{}
\mor{s1}{b4}{}
\mor{s2}{b4}{}
\mor{s3}{b4}{}
\mor{s5}{b4}{}
\mor{s6}{b4}{}
\mor{s10}{b4}{}
\mor{s1}{q1}{}[1,\dashline]
\mor{s2}{q2}{}[1,\dashline]
\mor{s3}{q3}{}[1,\dashline]
\mor{s4}{q4}{}[1,\dashline]
\mor{s5}{q5}{}[1,\dashline]
\mor{s6}{q6}{}[1,\dashline]
\mor{s7}{q7}{}[1,\dashline]
\mor{s8}{q8}{}[1,\dashline]
\mor{s9}{q9}{}[1,\dashline]
\mor{s10}{q10}{}[1,\dashline]
\mor{t1}{q1}{}[1,\dashline]
\mor{t2}{q2}{}[1,\dashline]
\mor{t3}{q3}{}[1,\dashline]
\mor{t4}{q4}{}[1,\dashline]
\mor{t5}{q5}{}[1,\dashline]
\mor{t6}{q6}{}[1,\dashline]
\mor{t7}{q7}{}[1,\dashline]
\mor{t8}{q8}{}[1,\dashline]
\mor{t9}{q9}{}[1,\dashline]
\mor{t10}{q10}{}[1,\dashline]
\enddc
$$
\caption{Computational graph for 1 level of Bini's \(\langle 3,2,2\rangle=10\) algorithm.  Solid lines indicate dependencies of additions and make up \(Enc_1A\), \(Enc_1B\), and \(Dec_1C\).  Dashed lines indicate dependencies of multiplications and connect these three subgraphs.  Note that \(Enc_1A\), the bottom-left part of the graph, is disconnected and has two connected components of equal size and equal input/output ratio.
Note that the base-case graph of Bini's algorithm is presented, for simplicity, with vertices of in-degree larger than two.  A vertex of degree larger than two, in fact, represents a full binary (not necessarily balanced) tree.  The expansion arguments hold for any way of drawing the binary trees.}
\label{fig:binigraph}
\end{figure}

From this \(\mnpq322{10}\) algorithm one immediately obtains 5 more algorithms by transposition and interchanging the encoding and decoding graphs \cite{HopcroftMusinsky73}.  Other algorithms can be constructed by taking tensor products of these base cases. When taking tensor products, the number of connected components of each encoding and decoding graph is the product of the number of connected components in the base cases.  For example there are 4 ways to construct algorithms for \(\mnpq664{100}\): one where \(Enc_1A\) and \(Enc_1B\) each have two components, one where \(Enc_1A\) and \(Dec_1C\) each have two components, one where \(Enc_1B\) and \(Dec_1C\) each have two components, and one where \(Enc_1A\) has four components.  Similarly there are 8 ways to construct algorithms for the square multiplication \(\mnpq{12}{12}{12}{1000}\).

\begin{table}[!ht]
\small
\begin{center}
  \begin{tabular}{ | c | c c | c | c | c | }
    \hline
    & Algorithm & Disconnected & Communication Cost Lower Bound & by & Tight? \\\hline\hline
    \multirow{14}{*}{\rotatebox{90}{Bini et.~al.~\cite{BiniCapovaniRomaniLotti79}}} 
    & \(\mnpq322{10}\) & \(EncA\) & $10^t/M^{\log_6 10-1}$ & Thm~\ref{thm:dec-con} & Yes\\\cline{2-6}
    & \multirow{2}{*}{\(\mnpq322{10}\)} & \multirow{2}{*}{\(DecC\)} & \(10^t/(t^{\log_6 10}M^{\log_6 10-1})\) & Thm~\ref{thm:enc-con} & Up to polylog factor\\
    & & & \(10^t/(M^{\log_3 5-1})\) & Cor~\ref{thm:dec-discon} & No\\\cline{2-6}
    & \multirow{2}{*}{\(\mnpq232{10}\)} & \multirow{2}{*}{\(EncA\)} & \(10^t/M^{\log_4 10-1}\) & Thm~\ref{thm:dec-con} & No\\
    & & & \(10^t/(t^{\log_6 10}M^{\log_6 10-1})\) & Thm~\ref{thm:enc-con} & Up to polylog factor\\\cline{2-6}
    & \multirow{2}{*}{\(\mnpq232{10}\)} & \multirow{2}{*}{\(EncB\)} & \(10^t/M^{\log_4 10-1}\) & Thm~\ref{thm:dec-con} & No\\
    & & & \(10^t/(t^{\log_6 10}M^{\log_6 10-1})\) & Thm~\ref{thm:enc-con} & Up to polylog factor\\\cline{2-6}
    & \(\mnpq223{10}\) & \(EncB\) & $10^t/M^{\log_6 10-1}$ & Thm~\ref{thm:dec-con} & Yes\\\cline{2-6}
    & \multirow{2}{*}{\(\mnpq223{10}\)} & \multirow{2}{*}{\(DecC\)} & \(10^t/(t^{\log_6 10}M^{\log_6 10-1})\) & Thm~\ref{thm:enc-con} & Up to polylog factor\\
    & & & \(10^t/(M^{\log_3 5-1})\) & Cor~\ref{thm:dec-discon} & No\\\cline{2-6}
    & \multirow{2}{*}{\(\mnpq664{100}\)} & \multirow{2}{*}{\(EncA,EncB\)} & \(100^t/M^{\log_{24}100-1}\) & Thm~\ref{thm:dec-con} & No\\
    & & & \(100^t/(t^{\log_{18}50}M^{\log_{18}50-1})\) & Cor~\ref{thm:enc-discon} & No\\\cline{2-6}
    & \(\mnpq{12}{12}{12}{1000}\) & \(EncA,EncB\) & \(1000^t/M^{\log_{144}1000-1}\) & \cite{BallardDemmelHoltzSchwartz11b} & Yes\\\hline\hline
    \multirow{10}{*}{\rotatebox{90}{Hopcroft-Kerr \cite{HopcroftKerr71}}}
    & \(\mnpq323{15}\) & None & \(15^t/M^{\log_9 15-1}\) & Thm~\ref{thm:dec-con} & Yes\\\cline{2-6}
    & \multirow{2}{*}{\(\mnpq332{15}\)} & \multirow{2}{*}{None} & \(15^t/M^{\log_6 15-1}\) & Thm~\ref{thm:dec-con} & No\\
    & & & \(15^t/(t^{\log_9 15}M^{\log_9 15-1})\) & Thm~\ref{thm:enc-discon} & Up to polylog factor\\\cline{2-6}
    & \multirow{2}{*}{\(\mnpq233{15}\)} & \multirow{2}{*}{None} & \(15^t/M^{\log_6 15-1}\) & Thm~\ref{thm:dec-con} & No\\
    & & & \(15^t/(t^{\log_9 15}M^{\log_9 15-1})\) & Thm~\ref{thm:enc-discon} & Up to polylog factor\\\cline{2-6}
    & \(\mnpq966{225}\) & None & \(225^t/M^{\log_{54} 225-1}\) & Thm~\ref{thm:dec-con} & Yes\\\cline{2-6}
    & \(\mnpq669{225}\) & None & \(225^t/M^{\log_{54} 225-1}\) & Thm~\ref{thm:dec-con} & Yes\\\cline{2-6}
    & \multirow{2}{*}{\(\mnpq696{225}\)} & \multirow{2}{*}{None} & \(225^t/M^{\log_{36} 225-1}\) & Thm~\ref{thm:dec-con} & No\\
    & & & \(225^t/(t^{\log_{54} 225}M^{\log_{54} 225-1})\) & Thm~\ref{thm:enc-discon} & Up to polylog factor\\\cline{2-6}
    & \(\mnpq{18}{18}{18}{3375}\) & None & \(3375^t/M^{\log_{324} 3375-1}\) & \cite{BallardDemmelHoltzSchwartz11b} & Yes\\\hline
  \end{tabular}
\end{center}
\caption{Asymptotic lower bounds for several variants of the algorithms by Bini et~al.~and Hopcroft-Kerr.  Many more with different shapes and with different disconnected subgraphs can be given for Bini's algorithm, and analyzed by similar means; we list only a representative sample.  Recall that the base case \(\mnpq mnpq\) is used for the computation of \(\mnpq{m^t}{n^t}{p^t}{q^t}\).}
\label{table:examples}
\end{table}

\subsection{The Hopcroft-Kerr algorithm}
Hopcroft and Kerr \cite{HopcroftKerr71} provide an algorithm for \(\mnpq323{15}\), and prove that fewer than 15 scalar multiplications is not possible.  In their algorithm, all the encoding and decoding graphs are connected.  Thus, only Theorems~\ref{thm:dec-con} and~\ref{thm:enc-con} are necessary for proving the lower bounds.  For the square case \(\mnpq{18}{18}{18}{3375}\), Theorem~\ref{thm:dec-con} reproduces the result of \cite{BallardDemmelHoltzSchwartz11b}.

\section{Discussion and Open Problems}
\label{sec:discussion}

Using graph expansion analysis we obtain tight lower bounds on recursive rectangular matrix multiplication algorithms in the case that the output matrix is at least as large as the input matrices, and the decoding graph is connected.  We also obtain a similar bound in the case that the encoding graph of the largest matrix is connected, which is tight up to a factor that is polylogarithmic in the input, assuming no multiply copied inputs.  Finally we extend these bounds to some disconnected cases, with restrictions on the fast memory size.  Whenever the decoding graph is not the largest of the three subgraphs (equivalently, whenever the output matrix is smaller than one of the input matrices), or when the largest graph is disconnected, our bounds are not tight.  

\subsection{Limitations of the lower bounds.}

There are several cases when our lower bounds do not apply.  These are cases where the full algorithm is a hybrid of several base algorithms combined in an arbitrary sequence.  Consider the case where two base algorithms are applied recursively.  If the recursion alternates between them, our lower bounds apply to the tensor product of the two base cases, which can be thought of as taking two recursive steps at once.   However, for cases of arbitrary choice of which base case to apply at each recursive step, we do not provide communication cost lower bounds.  The technical difficulty in extending our results in this case lies in generalizing the recursive construction of the decoding graph given in Section~\ref{sec:recconst}.  Similarly, if the base-case decoding (or encoding) graph is disconnected and contains several connected components of different sizes, our bounds do not apply.  In this case the connected components of the entire decoding (or encoding) graph are constructed out of all possible interleavings of the different connected components.  Finally, the lower bounds do not apply to algorithms that are not recursive, including approximate algorithms that are not bilinear.

\subsection{Parallel case.}
\label{sec:parallel}
Although our main focus is on the sequential case, we note that the sequential communication bounds presented here can be generalized to communication bounds in the distributed-memory parallel model of \cite{BallardDemmelHoltzLipshitzSchwartz12a}.  The lower bound proof technique here can be extended to obtain both memory-dependent and memory-independent parallel bounds as in \cite{BallardDemmelHoltzLipshitzSchwartz12b}.  Further, the Communication Avoiding Parallel Strassen \emph{(CAPS)} algorithm presented in \cite{BallardDemmelHoltzLipshitzSchwartz12a} is shown to be communication-optimal and faster (both theoretically and empirically) than previous attempts to parallelize Strassen's algorithm \cite{BallardDemmelLipshitzSchwartz12a}.  The parallelization approach of CAPS is general, and in particular it can be applied to rectangular matrix multiplication, giving a communication upper bound which matches the lower bounds in the same circumstances as in the sequential case.

\subsection{Blackbox use of fast square matrix multiplication algorithms.}
Instead of using a fast rectangular matrix multiplication algorithm, one can perform rectangular matrix multiplication of the form \(\langle m^t,n^t,p^t\rangle\) with fewer than the na\"ive number of \((mnp)^t\) multiplications by blackbox use of a square matrix multiplication algorithm with exponent \(\omega_0\) (that is, an algorithm for multiplying \(n\times n\) matrices with \(O(n^{\omega_0})\) flops). The idea is to break up the original problem into \(\left(\frac{m^t}{n^t}\right)\cdot\lt(\frac{p^t}{n^t}\rt)\) square matrix multiplication problems of size \((n^t)\times(n^t)\).\footnote{Assume, for simplicity, that \(n<m,p\).}  The arithmetic cost of such a blackbox algorithm is \(\Theta((mpn^{\omega_0-2})^t)\).  Using the upper and lower bounds in \cite{BallardDemmelHoltzSchwartz11b}, the communication cost is
\(\Theta\lt(\frac{(mpn^{\omega_0-2})^t}{M^{\omega_0/2-1}}\rt).\)

We note that, in some cases, blackbox use of a square algorithm may give a lower communication cost than a rectangular algorithm, even if it has a higher arithmetic cost.  In particular, if \(q<mpn^{\omega_0-2}\), then the rectangular algorithm performs asymptotically fewer flops.  It is possible to have simultaneously \(\omega_0/2>\log_{mp}q\), meaning that for certain values of \(M\) and \(t\) the communication cost of the rectangular algorithm is higher.  On some machines, the arithmetically slower algorithm may require less total time if the communication cost dominates.

\bibliographystyle{plain}
\bibliography{TAPAS}

\appendix

\section{Details of Bini's and the Hopcroft-Kerr algorithm}\label{app:details}
In this appendix we give the details of Bini's algorithm \cite{BiniCapovaniRomaniLotti79} and the Hopcroft-Kerr algorithm \cite{HopcroftKerr71}.  We provide these for completeness.

We express an algorithm for \(\mnpq mnpq\) matrix multiplication by giving the three adjacency matrices of the encoding and decoding graphs: \(U\) of dimension \(mn\times q\), \(V\) of dimension \(np\times q\), and \(W\) of dimension \(mp\times q\).  The rows of \(U\), \(V\), and \(W\), correspond to the entries of \(A\), \(B\), and \(C\), respectively, in row-major order.  The columns correspond to the \(q\) multiplications.  To be precise, each column of \(U\) specifies a linear combination of entries of \(A\); and each column of \(V\) specifies a linear combination of entries of \(B\).  These two linear combinations are to be multiplied together, and then the corresponding column of \(W\) specifies to which entries of \(C\) that product contributes, and with what coefficient.\footnote{The sparsity of the matrices in this notation correspond loosely to the number of additions and subtractions, but this notation is not sufficient to specify the leading constant hidden in the computational costs.  In particular, this notation does not show the advantage of Winograd's variant of Strassen's algorithm \cite{FischerProbert74} over Strassen's original formulation \cite{Strassen69}.}

\subsection{Bini's algorithm}
We provide all 6 base cases for Bini's algorithm that appear is Section~\ref{sec:examples:bini}.  They are labeled by the shape of the multiplication and which graph is disconnected.  The first algorithm is:
\[
U^{\langle3,2,2\rangle,EncA}=\left[\begin{array}{cccccccccc}
1& 0& 1& 0& 1& 0& 0& 0& 0& 0\\
0& 0& 0& \lambda& \lambda& 0& 0& 0& 0& 0\\
0& 0& 0& 0& 0& 1& 0& 1& 0& 1\\
1& 1& 0& 1& 0& 0& 0& 0& 0& 0\\
0& 0& 0& 0& 0& 0& 0& 0& \lambda& \lambda\\
0& 0& 0& 0& 0& 1& 1& 0& 1& 0
\end{array}\right]\equiv
\left[\begin{array}{c}
U_1\\U_2\\U_3\\U_4\\U_5\\U_6
\end{array}\right]
\] 
\[
V^{\langle3,2,2\rangle,EncA}=\left[\begin{array}{cccccccccc}
\lambda& 0& 0& -\lambda& 0& 1& 1& -1& 1& 0\\
0& 0& 0& 0& \lambda& 0& 0& -1& 0& 1\\
0& -1& 0& 1& 0& 0& 0& 0& \lambda& 0\\
1& -1& 1& 0& 1& \lambda& 0& 0& 0& -\lambda\\
\end{array}\right]\equiv
\left[\begin{array}{c}
V_1\\V_2\\V_3\\V_4
\end{array}\right]
\] 
\[
W^{\langle3,2,2\rangle,EncA}=\left[\begin{array}{cccccccccc}
\lambda^{-1}& \lambda^{-1}& -\lambda^{-1}& \lambda^{-1}& 0& 0& 0& 0& 0& 0\\
0& 0& -\lambda^{-1}& 0& \lambda^{-1}& 0& 0& 0& 0& 0\\
0& 0& 0& 1& 0& 1& 0& 0& -1& 0\\
1& 0& 0& 0& -1& 0& 0& 0& 0& 1\\
0& 0& 0& 0& 0& 0& -\lambda^{-1}& 0& \lambda^{-1}& 0\\
0& 0& 0& 0& 0& \lambda^{-1}& -\lambda^{-1}& \lambda^{-1}& 0& \lambda^{-1}
\end{array}\right]\equiv
\left[\begin{array}{c}
W_1\\W_2\\W_3\\W_4\\W_5\\W_6
\end{array}\right]
\] 
The remaining 5 algorithms can be concisely expressed in terms of the rows of the first algorithm:
\[U^{\langle3,2,2\rangle,DecC}=\left[\begin{array}{c}
W_1\\W_2\\W_3\\W_4\\W_5\\W_6
\end{array}\right]\qquad
V^{\langle3,2,2\rangle,DecC}=\left[\begin{array}{c}
V_1\\V_3\\V_2\\V_4
\end{array}\right]\qquad
W^{\langle3,2,2\rangle,DecC}=\left[\begin{array}{c}
U_1\\U_2\\U_3\\U_4\\U_5\\U_6
\end{array}\right]\qquad
\]

\[U^{\langle2,3,2\rangle,EncA}=\left[\begin{array}{c}
U_1\\U_3\\U_5\\U_2\\U_4\\U_6
\end{array}\right]\qquad
V^{\langle2,3,2\rangle,EncA}=\left[\begin{array}{c}
W_1\\W_2\\W_3\\W_4\\W_5\\W_6
\end{array}\right]\qquad
W^{\langle2,3,2\rangle,EncA}=\left[\begin{array}{c}
V_1\\V_2\\V_3\\V_4
\end{array}\right]\qquad
\]

\[U^{\langle2,3,2\rangle,EncB}=\left[\begin{array}{c}
W_1\\W_3\\W_5\\W_2\\W_4\\W_6
\end{array}\right]\qquad
V^{\langle2,3,2\rangle,EncB}=\left[\begin{array}{c}
U_1\\U_2\\U_3\\U_4\\U_5\\U_6
\end{array}\right]\qquad
W^{\langle2,3,2\rangle,EncB}=\left[\begin{array}{c}
V_1\\V_3\\V_2\\V_4
\end{array}\right]\qquad
\]

\[U^{\langle2,2,3\rangle,EncB}=\left[\begin{array}{c}
V_1\\V_3\\V_2\\V_4
\end{array}\right]\qquad
V^{\langle2,2,3\rangle,EncB}=\left[\begin{array}{c}
U_1\\U_3\\U_5\\U_2\\U_4\\U_6
\end{array}\right]\qquad
W^{\langle2,2,3\rangle,EncB}=\left[\begin{array}{c}
W_1\\W_3\\W_5\\W_2\\W_4\\W_6
\end{array}\right]\qquad
\]

\[U^{\langle2,2,3\rangle,DecC}=\left[\begin{array}{c}
V_1\\V_2\\V_3\\V_4
\end{array}\right]\qquad
V^{\langle2,2,3\rangle,DecC}=\left[\begin{array}{c}
W_1\\W_3\\W_5\\W_2\\W_4\\W_6
\end{array}\right]\qquad
W^{\langle2,2,3\rangle,DecC}=\left[\begin{array}{c}
U_1\\U_3\\U_5\\U_2\\U_4\\U_6
\end{array}\right]\qquad
\]

\subsection{The Hopcroft-Kerr algorithm}
For the Hopcroft-Kerr algorithm we give only 3 of the 6 base cases, since all the graphs are connected.
\[
U^{\langle 3,2,3\rangle}=\left[\begin{array}{ccccccccccccccc}
0& 1& 0& 1& 0& -1& 0& -1& 0& 0& 0& 0& 0& 0& -1\\
1& -1& 0& 0& 1& 0& 0& 1& 0& 1& 0& 0& 0& 0& 1\\
0& 0& 1& 1& 1& 0& 0& 0& 1& 0& 0& 0& -1& 1& 0\\
0& 0& 0& 0& 0& 0& 0& 0& -1& 1& -1& 0& 1& -1& 0\\
0& 0& 0& 0& 0& 1& 1& 1& 0& 0& 0& 0& 1& 0& 1\\
0& 0& 0& 0& 0& 0& 0& 0& 0& 0& 1& 1& -1& 1& -1
\end{array}\right]\equiv
\left[\begin{array}{c}
U_1\\U_2\\U_3\\U_4\\U_5\\U_6
\end{array}\right]
\] 

\[
V^{\langle 3,2,3\rangle}=\left[\begin{array}{ccccccccccccccc}
1& 1& 0& 1& 1& -1& 0& 1& 0& 0& 0& 0& 0& 0& 0\\
0& 0& 1& 1& 0& 0& 0& 0& 0& 0& -1& 0& -1& -1& 0\\
0& 0& 0& 0& 0& 1& 1& 0& & 0& 1& 0& 1& 1& 0\\
1& 0& 0& 0& 0& -1& 0& 1& 0& 1& 0& 0& 0& 0& -1\\
0& 0& 1& 0& -1& 0& 0& 0& 1& 1& -1& 0& 0& -1& 0\\
0& 0& 0& 0& 0& 1& 0& -1& 0& 0& 1& 1& 0& 0& 1
\end{array}\right]\equiv
\left[\begin{array}{c}
V_1\\V_2\\V_3\\V_4\\V_5\\V_6
\end{array}\right]
\] 

\[
W^{\langle 3,2,3\rangle}=\left[\begin{array}{ccccccccccccccc}
1& 1& 0& 0& 0& 0& 0& 0& 0& 0& 0& 0& 0& 0& 0\\
0& -1& -1& 1& -1& 0& 0& 0& 0& 0& 0& 0& 0& 0& 0\\
1& 0& 0& 0& 0& -1& 1& -1& 0& 0& 0& 0& 0& 0& 0\\
-1& 0& 0& 0& 1& 0& 0& 0& 1& 1& 0& 0& 0& 0& 0\\
0& 0& 1& 0& 0& 0& 0& 0& -1& 0& 0& 0& 0& 0& 0\\
0& 0& 1& 0& 0& 0& 0& 0& 0& 0& -1& 1& 0& 1& 0\\
0& 1& 0& 0& 0& 0& 0& 1& 0& 0& 0& 1& 0& 0& 1\\
0& 0& 0& 0& 0& 0& 1& 0& -1& 0& 0& 0& -1& -1& 0\\
0& 0& 0& 0& 0& 0& 1& 0& 0& 0& 0& 1& 0& 0& 0
\end{array}\right]\equiv
\left[\begin{array}{c}
W_1\\W_2\\W_3\\W_4\\W_5\\W_6\\W_7\\W_8\\W_9
\end{array}\right]
\] 

\[U^{\langle 2,3,3\rangle}=\left[\begin{array}{c}
U_1\\U_3\\U_5\\U_2\\U_4\\U_6
\end{array}\right]\qquad
V^{\langle 2,3,3\rangle}=\left[\begin{array}{c}
W_1\\W_2\\W_3\\W_4\\W_5\\W_6\\W_7\\W_8\\W_9
\end{array}\right]\qquad
W^{\langle 2,3,3\rangle}=\left[\begin{array}{c}
V_1\\V_2\\V_3\\V_4\\V_5\\V_6
\end{array}\right]\qquad
\]

\[U^{\langle 3,3,2\rangle}=\left[\begin{array}{c}
W_1\\W_2\\W_3\\W_4\\W_5\\W_6\\W_7\\W_8\\W_9
\end{array}\right]\qquad
V^{\langle 3,3,2\rangle}=\left[\begin{array}{c}
V_1\\V_4\\V_2\\V_5\\V_3\\V_6
\end{array}\right]\qquad
W^{\langle 3,3,2\rangle}=\left[\begin{array}{c}
U_1\\U_2\\U_3\\U_4\\U_5\\U_6
\end{array}\right]\qquad
\]

\end{document}